\documentclass[12pt]{iopart}

\usepackage{lineno,hyperref}
\modulolinenumbers[5]
\usepackage[utf8]{inputenc}
\usepackage{xcolor}
\usepackage{subfig}
\usepackage{graphicx}
\usepackage{amsthm}
\usepackage{longtable}
\usepackage{multirow}
\usepackage{epstopdf}

\newtheorem{theorem}{Theorem}
\newtheorem{remark}{Remark}
\begin{document}

\title{Endemic state equivalence between non-Markovian SEIS and Markov SIS model in complex networks}

\author[address1]{Igor Tomovski$^{1,*}$, Lasko Basnarkov$^{2,1}$, Alajdin Abazi$^{1,3}$}
\address{$^1$Research Center for Computer Science and Information Technologies, Macedonian Academy of Sciences and Arts, Bul. Krste Misirkov, 2, P.O. Box 428, 1000 Skopje, Macedonia }

\address{$^2$Faculty of Computer Science and Engineering, "Ss Cyril and Methodius" University - Skopje,  ul.Rudzer Boshkovikj 16, P.O. Box 393,1000 Skopje, Macedonia}
\address{$^3$South East European University, Ilindenska n.335, 1200 Tetovo, Macedonia}
\address{$^*$Author to whom any correspondence should be addressed.}
\ead{igor@manu.edu.mk}

\vspace{10pt}
\begin{indented}
\item[]
\end{indented}

\begin{abstract}
In the light of several major epidemic events that emerged in the past two decades, and emphasized by the COVID-19 pandemics, the non-Markovian spreading models occurring on complex networks gained significant attention from the scientific community. Following this interest, in this article, we explore the relations that exist between the non-Markovian SEIS (Susceptible--Exposed--Infectious--Susceptible) and the classical Markov SIS, as basic re-occurring virus spreading models in complex networks. We investigate the similarities and seek for equivalences both for the discrete-time and the continuous-time forms. First, we formally introduce the continuous-time non-Markovian SEIS model, and derive the epidemic threshold in a strict mathematical procedure. Then we present the main result of the paper that, providing certain relations between process parameters hold, the stationary-state solutions of the status probabilities in the non-Markovian SEIS  may be found from the stationary state probabilities of the Markov SIS model. This result has a two-fold significance. First, it simplifies the computational complexity of the non-Markovian model in practical applications, where only the stationary distribution of the state probabilities is required. Next, it defines the epidemic threshold of the non-Markovian SEIS model, without the necessity of a thrall mathematical analysis. We present this result both in analytical form, and confirm the result trough numerical simulations. Furthermore, as of secondary importance, in an analytical procedure we show that each Markov SIS may be represented as non-Markovian SEIS model.
\end{abstract}

\noindent{\it Keywords}: Complex networks, Epidemic models, non-Markovian processes, Endemic states, Stability analysis
%\maketitle

\section{Introduction}

Following the outbreak of several recent epidemics, the SARS, the MERS, Bird flu etc., and emphasized by the Covid-19 pandemics, the non-Markovian models captured the attention of the complex networks research community that studies stochastic spreading processes  \cite{starnini, Nowzari, boguna1, delft_nm1, delft_nm2, delft_nm3, kiss1, kiss2, kiss3, Feng2019}. The shift of interest from the Markov to the non-Markovian realm was caused by the realization that no status transition, that an individual (node) undergoes following the contraction of the spread agent, may neither occur simultaneously, nor the probability of status transition on a daily bases is constant, as may be seen from the collected medical data \cite{chun,Qineabc1202}. For example, when an individual contracts a virus, the individual has no capacity to instantly spread the agent to the neighboring nodes. A certain viral quantity should first be produced (the node acts as an incubator for the agent), in order for the infected individual to become infectious. The time required for the virus to replicate over the necessary viral load threshold, varies from individual to individual and follows a time dependent probability distribution \cite{He2020}. On the other hand, the recovery process is a product of complex biochemical interactions within the hosts immune system, that take several days in order for proper deference response to be prepared and for the virus to be eradicated (as an example of the immune system modeling one may refer to \cite{Perelson2002}). Again, individual differences lead to specific probability distribution of time from Exposure to Recovery.

In our recent paper \cite{nmseis}, we introduced the non-Markovian SEIS (Susceptible -- Exposed -- Infectious -- Susceptible) model as a basic mathematical non-Markovian form that describes re-occurring spreading processes, taking place on complex networks. In the model formulation, We assumed that status transitions from Exposed (non-Infectious) to Infectious status and  Exposed (both non-Infectious and Infectious) back to Susceptible status, follow temporal distribution described with Discrete Time Probability Functions (DTPFs): 
\begin{itemize}
\item daily manifesting function $b(\tau)$: probability that an Exposed and previously non-Infectious node, becomes Infectious exactly at day $\tau$;
\item manifesting function $B(\tau)$: probability that an Exposed node, is Infectious at day $\tau$;
\item daily recovering function $\gamma(\tau)$: probability that an Exposed node, recovers exactly at day $\tau$;
\item recovering function $\Gamma(\tau)$: probability that an Exposed node, is recovered by day $\tau$;
\end{itemize}

In this paper, We first extend the discrete-time concept to continuous-time non-Markovian model form. Adequately, the functions $\gamma(\tau)$, $\Gamma(\tau)$, $b(\tau)$ and $B(\tau)$ in this scenario are continuous, and are further referred to as Continuous-Time Probability Functions (CTPFs), with $\gamma(\tau)$ and $b(\tau)$ referred to as instance recovering probability and instance manifesting probability, correspondingly.  For the continuous-time form, we derive the epdemic threshold in a strict mathematical procedure. Then, the main result of this paper is presented: that for each non-Markovian SEIS model (discrete-time or continuous-time), exists a Markov SIS model, such that the stationary state probabilities of each node being exposed in the SEIS model, equals the stationary probability that the node is Infected in the SIS model, providing certain relations between process parameters hold. We consider this result to be of at-most importance for the following reason: non-Markovian models, although highly accurate in analyzing natural phenomena, are computationally sufficiently more demanding. Investigating these models with utilization of Markov analogs (as shown here as possible), significantly reduces the computational complexity in acquiring significant data related to the endemic state of the diseases. The presented analysis directly leads to relations that define the epidemic threshold for the non-Markovian (SEIS) models occurring on complex networks, without the necessity of a thrall mathematical procedure. Similar type of equivalences, between the non-Markovian and the Markov SIS model, using different settings and approaches, have been established by the authors in \cite{starnini} and \cite{Feng2019}

As a result of secondary importance, it is shown that an arbitrary Markov SIS model occurring on complex networks, may be represented an non-Markovian SEIS model. This equivalence is only vaguely mentioned and numerically illustrated in the Conclusions of \cite{nmseis}; here we show this feature trough a rigorous mathematical procedure. One should note that similar analysis was conducted in respect to the non-Markovian and Markov SIR model in \cite{nmsir}.

\section{The model}
The discrete-time form of the SEIS model analyzed in this paper is originally introduced in \cite{nmseis}. For completeness, in what follows, we re-state the formal definition of the observed process, and for more details we refer the readers to the cited paper. 

Consider a network represented with the adjacency matrix $\mathbf{A}$. In the general case, the network is directed, weighted, and strongly connected; consequently the matrix $\mathbf{A}=[a_{ij}]$ is asymmetric, with $0 \le a_{ij} \le 1$ and irreducible (Perron-Frobenious theorem for non-negative irreducible matrices applies). 

The SEIS model is a status model in which, in respect to the spreading agent, each node is in one of the three following statuses: Susceptible, Exposed and Infectious. A node is in status Exposed, at time $t$, if at the given instance it contains the spread agent. Exposed node may be Infectious (manifesting infectiousness) or non-Infectious. Node is Infectious if it contains the agent (is Exposed) and is capable to spread the agent to the neighboring nodes.

When Susceptible node contract the spread agent, the node becomes Exposed (and generally assumed non-Infectious). The process of agent contraction by the Susceptible node plays a role of a trigger event ($\tau=0$): all consequent processes within the node are time-referenced to this transition. Exposed (but non-Infectious node) may become Infectious exactly at time $\tau$ after the trigger event with probability $b(\tau)$. Exposed node is Infectious at time $\tau$ following the trigger event with probability $B(\tau)$. To stress the difference between $b(\tau)$ and $B(\tau)$, as explained in \cite{nmseis}, the model allows for two different types of Infectiousness manifestation:
\begin{itemize}
	\item Cumulative manifestation -- in this case $b(\tau)$ has a character of a mass probability function in the discrete-time case scenario and density probability function in the continuous-time scenario. Adequately, $B(\tau)=\sum_{k=0}^{\tau}b(k)$ (discrete-time), with $\sum_{k=0}^{T-1} b(\tau) \le 1$,  or $B(\tau)=\int_{0}^{\tau}b(\tau')d \tau'$ (continuous-time), with $\int_{0}^{T}b(\tau) d\tau \le 1$, has a cumulative character, with the sign "$<$" indicating that the Exposed node may not necessarily become Infectious prior to recovery. This type of behavior is typical for epidemic diseases;
	\item Random manifestation -- in this case $B(\tau)=b(\tau)$ has a random character, with $0 \le b(\tau) \le 1$ being the only restriction.
\end{itemize}
Exposed node recovers and becomes Susceptible again exactly at time $\tau$ following the exposure, with probability $\gamma(\tau)$, and is recovered at time $\tau$ with probability $\Gamma(\tau)=\sum_{k=0}^{\tau} \gamma(k)$, with $\sum_{\tau=0}^{T-1} \gamma(\tau)=1$, in the discrete-time model, and  $\Gamma(\tau)=\int_{0}^{\tau} \gamma(\tau') d\tau'$, with $\int_{0}^{T} \gamma(\tau) d\tau=1$ in the conitinuous case. The probability function (DTPF/CTPF) $\gamma(\tau)$ is a m.p.f in the discrete-time scenario, and p.d.f. in the continuous-time case, with $\Gamma(\tau)$ being a cumulative probability function. In the modeling of the SEIS process, we widely use the complement $\overline{\Gamma}(\tau)=1-\Gamma(\tau)$.

In the formal sense, the dynamical behavior of the model is defined as follows: node $i$ is Exposed at time $t$ if it contracted the agent at time $t-\tau$ and did not recover in the time interval $[t-\tau,t]$;  node $i$ is Infectious at time $t$ if it contracted the agent at time $t-\tau$, did not recover in the time interval $[t-\tau,t]$ and is capable to spread the agent to its neighbors at time $t$. We consider that the process lasts for maximum $T$ time units, with $\Gamma(T-1)=1$, in the discrete-time scenario, and $\Gamma(T)=1$, in the continuous-time version.

Considering the definitions stated above, and following \cite{nmseis}, the discrete-time SEIS model is mathematically defined in the following form:
\begin{eqnarray}
&&p^E_i(t+1)=\sum_{\tau=0}^{T-1}(1-p^E_i(t-\tau)) \overline{\Gamma} (\tau) \mathcal{P}_i(t-\tau)\label{disSEIS}\\
&&p^I_i(t+1)=\sum_{\tau=0}^{T-1} (1-p^E_i(t-\tau)) B (\tau) \overline{\Gamma} (\tau) \mathcal{P}_i(t-\tau), \nonumber
\end{eqnarray}

with $\mathcal{P}_i(t)$ representing the product-like term:
\begin{eqnarray}
&&\mathcal{P}_i(t)=1-\prod_{j=1}^N (1-p^I_j(t)  a_{ij} \beta), \nonumber
\end{eqnarray}
that denotes the probability that a Susceptible node $i$ will contract the spread agent from its neighbours, at time $t$ \cite{wang,Chakrabarti,gomez1,gomez2}.

\subsection{Continuous-time SEIS model}

In this section, we introduce the continuous-time non-Markovian SEIS model, as an extension to the discrete-time model (\ref{disSEIS}). Similar model forms, represented as non-Markovian SIS models may be found in \cite{starnini, delft_nm1,delft_nm2,delft_nm3, Feng2019}. The difference between the approach we take in our formulation and the approach in the cited papers, relates to the process of the transfer of the infectious material from the Exposed/Infected node to its neighbors. In the cited papers the non-Markovian character is expressed in the form of time-distributed infection rate $\beta(\tau)$. In our approach, the non-Markovian feature lays within the capability of the Exposed node to transfer the agent to its neighbors (non-Markovinity of manifestation), while transfer itself is of Markov type. 

Prior to the formal introduction of the continuous-time SEIS model, we stress the major difference that exist between discrete-time and continuous-time modeling approach. It is a standard practice in modeling spreading phenomena in continuous time to assume that the transfer of the spread agent, within an infinitesimal time interval $\Delta \tau$, may occur from a single sources (neighbour), only (no-multiple infectious events assumption). This notion transforms the product-like term in the following manner:
\begin{eqnarray}
1-\prod_{j=1}^N (1-p^I_j(t-\tau)  a_{ij} \beta \Delta \tau)&=&\sum p^I_j(t-\tau)  a_{ij} \beta\Delta \tau +O(\Delta \tau^2). \nonumber
\end{eqnarray}
For sufficiently small $\Delta \tau$, the term $O(\Delta \tau^2)$ is neglected, and the appropriate sum-like term \cite{delft}, obtained.

Bearing in mind the differences, We may now re-write the system of equations (\ref{disSEIS}) as follows:
\begin{eqnarray}
p^E_i(t)&=&\sum_{k=1}^{T/(\Delta \tau)}(1-p^E_i(t-k \Delta \tau)) \sum_{j=1}^N p^I_j(t-k \Delta \tau) \overline{\Gamma} ((k-1) \Delta \tau) a_{ij} \beta\Delta \tau \nonumber\\ 
&=&\sum_{k=1}^{T/(\Delta \tau)}(1-p^E_i(t-k \Delta \tau)) \sum_{j=1}^N p^I_j(t-k \Delta \tau) \overline{\Gamma} (k \Delta \tau) a_{ij} \beta\Delta \tau -\nonumber\\
&-&\sum_{k=1}^{T/(\Delta \tau)}(1-p^E_i(t-k \Delta \tau)) \sum_{j=1}^N p^I_j(t-k \Delta \tau) \overline{\Gamma}' (k \Delta \tau) a_{ij} \beta\Delta \tau^2 \label{contsum}\\
p^I_i(t)&=&\sum_{k=1}^{T/(\Delta \tau)} (1-p^E_i(t-k \Delta \tau)) \times \nonumber\\ &\times& \sum_{j=1}^N p^I_j(t-k \Delta \tau) B ((k-1) \Delta \tau) \overline{\Gamma} ((k-1) \Delta \tau) a_{ij} \beta\Delta \tau \nonumber\\
&=&\sum_{k=1}^{T/(\Delta \tau)}(1-p^E_i(t-k \Delta \tau)) \sum_{j=1}^N p^I_j(t-k \Delta \tau) B(k \Delta \tau) \overline{\Gamma} (k \Delta \tau) a_{ij} \beta\Delta \tau -\nonumber\\
&-&\sum_{k=1}^{T/(\Delta \tau)}(1-p^E_i(t-k \Delta \tau)) \sum_{j=1}^N p^I_j(t-k \Delta \tau) [B(k \Delta \tau)\overline{\Gamma} (k \Delta \tau)]' a_{ij} \beta\Delta \tau^2 \nonumber
\end{eqnarray}

When $\Delta \tau \rightarrow 0$, providing no discontinuities of first kind exist in $\overline{\Gamma}(\tau)$ or $\overline{\Gamma}(\tau)B(\tau)$, the terms multiplied by $\Delta \tau^2$ may be neglected. In what follows, we show that this term may be neglected even in the presence of finite number of first order discontinuities. The analyses is focused around the second set of $N$ equation in the system (\ref{contsum}), related to the $p^I_i(t)$, variables; by analogy, the same analysis is valid for the set of equations related to the $p^E_i(t)$ variables. 

Consider a point $0 \le \tau_i <T$, such that a discontinuity of first kind $\overline{\Gamma}(\tau)$ or $\overline{\Gamma}(\tau)B(\tau)$ exists at $\tau=\tau_i$. Let $\Delta \tau$ be an integration constant, such that the series $\tau_k=k \Delta \tau$ provides a proper sampling of $\overline{\Gamma}(\tau)$ and $\overline{\Gamma}(\tau)B(\tau)$. Let $k_i$ be an index such that $k_{i-1} \Delta \tau < \tau_i \le k_i \Delta \tau$. Under these assumptions, the following inequality may be considered:
\begin{eqnarray}
&&|(1-p^E_i(t-k_i \Delta \tau)) \sum_{j=1}^N p^I_j(t-k_i \Delta \tau) [B(k_i \Delta \tau)\overline{\Gamma} (k_i \Delta \tau)]' a_{ij} \beta\Delta \tau^2| \le \nonumber\\
&& |(1-p^E_i(t-k_i \Delta \tau))| \sum_{j=1}^N |p^I_j(t-k_i \Delta \tau)| |[B(k_i \Delta \tau)\overline{\Gamma} (k_i \Delta \tau)]'| a_{ij} \beta\Delta \tau^2 \approx \nonumber\\
&& |(1-p^E_i(t-k_i \Delta \tau))| \times \nonumber \\
&& \times \sum_{j=1}^N |p^I_j(t-k_i \Delta \tau)| \frac{|B(k_i \Delta \tau)\overline{\Gamma}(k_i \Delta \tau) - B(k_{i-1} \Delta \tau)\overline{\Gamma}(k_{i-1} \Delta \tau)|} {\Delta \tau} a_{ij} \beta\Delta \tau^2 \nonumber\\ &&\le N \Delta \tau \nonumber
\end{eqnarray}

Let R be a total number of first kind discontinuities of either $\overline{\Gamma}(\tau)$ and $\overline{\Gamma}(\tau)B(\tau)$. Then, in accordance with the relation above:

\begin{eqnarray}
&&|\sum_{k=1}^{T/(\Delta \tau)}(1-p^E_i(t-k \Delta \tau)) \sum_{j=1}^N p^I_j(t-k \Delta \tau) [B(k \Delta \tau)\overline{\Gamma} (k \Delta \tau)]' a_{ij} \beta\Delta \tau^2| \le \nonumber\\
&&|\sum_{k \ne k_1,.,k_R}(1-p^E_i(t-k \Delta \tau)) \sum_{j=1}^N p^I_j(t-k \Delta \tau) [B(k \Delta \tau)\overline{\Gamma} (k \Delta \tau)]' a_{ij} \beta \Delta \tau^2| \nonumber\\ &&+  NR \Delta \tau\nonumber
\end{eqnarray}

The preceding analyses indicates that the the terms multiplied by $\Delta \tau^2$ in the set of $N$ equations, related to the $p^I_i(t)$ in (\ref{contsum}), may be neglected, since for finite $R$, $NR \Delta \tau$ may be maid arbitrary small, with the right choice of $\Delta \tau$. Similar analysis, leading to the same conclusion, may be conducted for the set of equations related to $p^E_i(t)$ in (\ref{contsum}). Consequently, from (\ref{contsum}) and considering $\Delta \tau \rightarrow 0$, one obtains the integral form of equations for the non-Markovian SEIS model occurring on complex networks in continuous-time: 

\begin{eqnarray}
p^E_i(t)&=&\int_{0}^{T}(1-p^E_i(t-\tau)) s(\tau) \sum_{j=1}^N p^I_j(t- \tau) \overline{\Gamma} (\tau) a_{ij} \beta d\tau \label{contsys}\\
p^I_i(t)&=&\int_{0}^{T} (1-p^E_i(t-\tau)) s(\tau) \sum_{j=1}^N p^I_j(t-\tau) B (\tau) \overline{\Gamma} (\tau) a_{ij} \beta d\tau\nonumber
\end{eqnarray}
with $s(t)$ being the Heaviside function. In what follows we consider both $\Gamma(\tau)$ and $B(\tau)$ to be smooth around the point $\tau=0$, and the Heaviside function may be neglected in the system of equations (\ref{contsys}).

\subsubsection{Differential form}

In this segment, we show that the non-Markovian SEIS model, represented with (\ref{contsys}), may be written in a differential form, as well. The purpose of this model-form is to relate the non-Markovian SEIS model and the Markov SIS model, in order to investigate the circumstances under which an arbitrary Markov SIS may be presented as non-Markovian SEIS.

Starting from the system of equations (\ref{contsum}), one obtains:
\begin{eqnarray}
&&p^E_i(t+\Delta \tau)-p^E_i(t)= (1-p^E_i(t)) \sum_{j=1}^N p^I_j(t) \overline{\Gamma} (0) a_{ij} \beta\Delta \tau + \nonumber\\ 
&&+\sum_{k=1}^{T/(\Delta \tau)} (1-p^E_i(t-k \Delta \tau)) \sum_{j=1}^N p^I_j(t-k \Delta \tau) [\overline{\Gamma} (k \Delta \tau) -  \overline{\Gamma} ((k-1) \Delta \tau)]a_{ij} \beta\Delta \tau \nonumber\\
&&- (1-p^E_i(t-T- \Delta \tau)) \sum_{j=1}^N p^I_j(t-T-\Delta \tau) \overline{\Gamma} (T) a_{ij} \beta\Delta \tau \nonumber
\end{eqnarray}

Considering that $\overline{\Gamma}(T)=0$, by dividing both sides of the equation with $\Delta \tau$ and by letting $\Delta \tau \rightarrow 0$, the following relation may be written :
\begin{eqnarray}
\frac{dp^E_i(t)}{dt}&=&  (1-p^E_i(t))  \sum_{j=1}^N p^I_j(t) \overline{\Gamma} (0) a_{ij} \beta +\nonumber\\ &+& \int_{0}^{T} s(\tau) (1-p^E_i(t-\tau))  \sum_{j=1}^N p^I_j(t-\tau) a_{ij} \beta \overline{\Gamma}' (\tau)  d\tau\label{dif1}
\end{eqnarray}

Similarly for $p^I_i$ one obtains:

\begin{eqnarray}
\frac{dp^I_i(t)}{dt}&=& (1-p^E_i(t)) \sum_{j=1}^N p^I_j(t) \overline{\Gamma} (0) B(0) a_{ij} \beta +\nonumber\\ &+& \int_{0}^{T}  s(\tau) (1-p^E_i(t-\tau)) \sum_{j=1}^N p^I_j(t-\tau) a_{ij} \beta [\overline{\Gamma} (\tau)B(\tau)]'  d\tau\label{dif2}
\end{eqnarray}

\subsubsection{Epidemic threshold for the continuous-time model}

One of the main results that are derived in the theoretical analysis of the stochastic spreading processes occurring on complex networks, is the determination of the epidemic threshold. Epidemic threshold defines the critical relation between the process parameters and network topology, that separate the parametric region in which the network is disease free, from the region in which a permanent epidemic exists.

To find the epidemic threshold for the continuous-time non-Markovian SEIS model, we resort to the investigation of the stability criteria of the dynamical system (\ref{contsys}), around the point of epidemic origin, i.e. $p^E_i(t)=0$, $p^I_i(t)=0$, for all $i$. In that sense, we consider the following:

\begin{theorem}
Consider a directed, weighted and strongly connected graph, represented with the adjacency matrix $\mathbf{A}=[a_{ij}]$ that is, consequently, non-negative and irreducible. The $2N$ vector $[p^E_i(t), p^I_i(t)]=[0,0]$, i.e. the epidemic origin, is a globally asymptotically stable point of equilibrium of the dynamical system (\ref{contsys}), providing the following relation holds:
\begin{eqnarray}
\frac{1}{\beta \lambda_1(\mathbf{A})} > \int_{0}^{T}B (\tau) \overline{\Gamma}(\tau)d\tau \nonumber
\end{eqnarray}
with $\lambda_1(\mathbf{A})$ being the leading eigenvalue of the matrix $\mathbf{A}$. 
\end{theorem}

\begin{proof}
Consider the system (\ref{contsys}). Since $p^E_i(t), p^I_i(t) \in [0,1]$, for all $i$, and $B(\tau), \overline{\Gamma}(\tau) \ge 0$, the argument under the integral is strictly postitve, so the following relation holds:
\begin{eqnarray}
p^E_i(t)&\le&\int_{0}^{T} \sum_{j=1}^N p^I_j(t- \tau) \overline{\Gamma} (\tau) a_{ij} \beta d\tau \nonumber\\
p^I_i(t)&\le&\int_{0}^{T} \sum_{j=1}^N p^I_j(t-\tau) B (\tau) \overline{\Gamma} (\tau) a_{ij} \beta d\tau\nonumber
\end{eqnarray}
In other words, the dynamical behaviour of the system (\ref{contsys}) is bounded from bellow by the epidemic origin, and from above, by the dynamical system:
\begin{eqnarray}
p'^E_i(t)&=&\int_{0}^{T} \sum_{j=1}^N p'^I_j(t- \tau) \overline{\Gamma} (\tau) a_{ij} \beta d\tau \label{auxsys}\\
p'^I_i(t)&=&\int_{0}^{T} \sum_{j=1}^N p'^I_j(t-\tau) B (\tau) \overline{\Gamma} (\tau) a_{ij} \beta d\tau\nonumber
\end{eqnarray}

Consequently, if $lim_{t \rightarrow \infty} p'^E_i(t) \rightarrow 0$, $lim_{t \rightarrow \infty} p'^I_i(t) \rightarrow 0$, then $lim_{t \rightarrow \infty} p^E_i(t) \rightarrow 0$, $lim_{t \rightarrow \infty} p^I_i(t) \rightarrow 0$, as well. In that sense, the proof of the global stability of the (epidemic origin of the) system (\ref{contsys}), reduces to proof of the global stability of the system (\ref{auxsys}).

One may notice that the second set of $N$ equations in (\ref{auxsys}) is self-sufficient. In that sense the dynamical stability of the system (\ref{auxsys}) reduces to the dynamical stability of this set of equations only: from (\ref{auxsys}) follows that, if $lim_{t \rightarrow \infty} p'^I_i(t) \rightarrow 0$, then $lim_{t \rightarrow \infty} p'^E_i(t) \rightarrow 0$.

By conducting a Laplace transform on both sides of each of the equations from the second set of $N$ equations in (\ref{auxsys}), following the methodology in \cite{nmsir}, one obtains:
\begin{eqnarray}
&&P'^I_i(s)= \int_{0}^{\infty} p'^I_i(t)e^{-st}dt =\int_{0}^{\infty}e^{-st} \int_{0}^{T} \sum_{j=1}^N p'^I_j(t-\tau) B (\tau) \overline{\Gamma} (\tau) a_{ij} \beta d\tau dt\nonumber =\nonumber\\
&&= \int_{0}^{T} B (\tau) \overline{\Gamma}(\tau)  e^{-s\tau} \beta \left( \sum_{j=1}^N a_{ij} \int_{-T}^{\infty}p'^I_j(u) e^{-su}du \right) d\tau = \label{LaplEqn}\\
&&=\beta \mathcal{L}(B (\tau) \overline{\Gamma}(\tau)) \sum_{j=1}^N a_{ij}\left( P'^I_j(s) + \int_{-T}^{0}p'^I_j(u) e^{-su}du \right)= \nonumber\\
&&= \beta \mathcal{L}(B (\tau) \overline{\Gamma}(\tau)) \sum_{j=1}^N a_{ij}\left( P'^I_j(s) + P'^I_j(0^-,s), \right)\nonumber
\end{eqnarray}
with
\begin{eqnarray}
L(S)=\mathcal{L}(B (\tau) \overline{\Gamma}(\tau)) = \int_0^{\infty} B (\tau) \overline{\Gamma}(\tau) e^{-s\tau}d\tau=\int_0^T B (\tau) \overline{\Gamma}(\tau) e^{-s\tau}d\tau, \label{BGlapl}
\end{eqnarray}
being the Laplace transform of the product $B (\tau) \overline{\Gamma}(\tau)$, and
\begin{eqnarray}
P'^I_j(0^-,s)=\int_{-T}^{0}p'^I_j(u) e^{-su}du \nonumber
\end{eqnarray}
the term that takes into account the initial conditions. The initial conditions and the nature of this term are discussed separately in Remark 1, at the end of this segment. 

Consider the vector $\mathbf{P'} ^I (s)=[P'^I_i(s)]$ and $\mathbf{P'} ^I (0,s)=[P'^I_i(0,s)]$. The system of equation (\ref{LaplEqn}) may be re-written in a vector form as:
\begin{eqnarray}
\mathbf{P'} ^I(s)= \beta \mathcal{L}(B (\tau) \overline{\Gamma}(\tau)) \mathbf{A}(\mathbf{P'} ^I(s)+\mathbf{P'} ^I(0,s)), \label{LaplVect}
\end{eqnarray}
leading to a solution in Laplace domain, in the following form:
\begin{eqnarray}
\mathbf{P'} ^I(s)&=& \beta \mathcal{L}(B (\tau) \overline{\Gamma}(\tau)) \mathbf{A} \mathbf{P'} ^I(0,s) (\mathbf{I} - \beta \mathcal{L}(B (\tau) \overline{\Gamma}(\tau) \mathbf{A})^{-1}= \nonumber\\
&=&\beta \mathcal{L}(B (\tau) \overline{\Gamma}(\tau)) \mathbf{A} \mathbf{P'} ^I(0,s) \frac{(\mathbf{I} - \beta \mathcal{L}(B (\tau) \overline{\Gamma}(\tau)) \mathbf{A})'}{det(\mathbf{I} - \beta \mathcal{L}(B (\tau) \overline{\Gamma}(\tau)) \mathbf{A})} \label{LapSol}
\end{eqnarray}
where $(\mathbf{I} - \beta \mathcal{L}(B (\tau) \overline{\Gamma}(\tau)) \mathbf{A})'$ is a matrix, which elements are the minors of the matrix $\mathbf{I} - \beta \mathcal{L}(B (\tau) \overline{\Gamma}(\tau)) \mathbf{A}$.

It is a well known result in the dynamical system theory, that the stability of the (origin of the) dynamical system  is determined by the position of the poles of the system in the complex plane. If all poles of the dynamical system lie within the left-half of the complex plane, i.e. $Re\{s\}<0$, the dynamical system is globally asymptotically stable. From the equation (\ref{LapSol}), one obtains that the poles of the system (\ref{auxsys}) may be determined from the zeroes of the equation:
\begin{eqnarray}
&&det(\mathbf{I} - \beta \mathcal{L}(B (\tau) \overline{\Gamma}(\tau)) \mathbf{A})=0 \nonumber
\end{eqnarray}
   
On the other hand: 
\begin{eqnarray}
det(\mathbf{I} - \beta \mathcal{L}(B (\tau) \overline{\Gamma}(\tau)) \mathbf{A})=\left(\beta \mathcal{L}(B (\tau) \overline{\Gamma}(\tau))\right)^N det\left(\frac{1}{\beta \mathcal{L}(B (\tau) \overline{\Gamma}(\tau))}\mathbf{I} -  \mathbf{A}\right)= \nonumber \\
=\left(\beta \mathcal{L}(B (\tau) \overline{\Gamma}(\tau)\right)^N \prod_{i=1}^N \left(\frac{1}{\beta \mathcal{L}(B (\tau) \overline{\Gamma}(\tau))} - \lambda_i(\mathbf{A})\right)= \prod_{i=1}^N \left(1 - \lambda_i(\mathbf{A}) \beta \mathcal{L}(B (\tau) \overline{\Gamma}(\tau))\right),\nonumber
\end{eqnarray}
with $\lambda_i(\mathbf{A})$, $i=1,..,N$, being the eigenvalues of the adjacency matrix $\mathbf{A}$.
From the last relation, the position of the poles of the dynamical system (\ref{auxsys}) are determined from the set of equations:
\begin{eqnarray}
&&1 - \lambda_i(\mathbf{A}) \beta \mathcal{L}(B (\tau) \overline{\Gamma}(\tau))=0 \label{final_solution}
\end{eqnarray}

Let $s_{i,k}$, be a pole of the dynamical system (\ref{auxsys}), associated with the $i$-th eigenvalue $\lambda_i(\mathbf{A})$ in the following manner:
\begin{eqnarray}
&&L(s_{i,k})=\mathcal{L}(B (\tau) \overline{\Gamma}(\tau)) \Big|_{s=s_{i,k}} = \int_0^T B (\tau) \overline{\Gamma}(\tau)e^{-s_{i,k}\tau}d\tau= \frac{1}{\beta \lambda_i(\mathbf{A})} \nonumber
\end{eqnarray}
Index $k$ allows for multiple poles associated with a single eigenvalue $\lambda_i(\mathbf{A})$.

From the definition of the Laplace transform of $\mathcal{L}(B (\tau) \overline{\Gamma}(\tau))$, i.e. equation (\ref{BGlapl}), the following conclusions hold:
\begin{itemize}
    \item if $Re\{s_{i,k}\}>0$, i.e. if a pole of the dynamical system (\ref{auxsys}) lies on the right-half of the complex plane, the value of the term $L(s_{i,k})$ lies within the circle $|z|=\int_{0}^{T}B (\tau) \overline{\Gamma}(\tau)d\tau$ in the complex plain;
    \item if $Re\{s_{i,k}\}<0$, i.e. if a pole of the dynamical system (\ref{auxsys}) lies on the left-half of the complex plane, the value of the term $L(s_{i,k})$ lies outside the circle $|z|=\int_{0}^{T}B (\tau) \overline{\Gamma}(\tau)d\tau$;
\end{itemize}
Bearing in mind the preceding discussion and the relation (\ref{final_solution}), the poles of the dynamical system (\ref{auxsys}) will lie within the left-half of the complex plane, i.e. $Re\{s_{i,k}\}<0$, providing:
\begin{eqnarray}
&& ||\mathcal{L}(B (\tau) \overline{\Gamma}(\tau))|| \big|_{s=s_{i,k}}= \frac{1}{\beta ||\lambda_i(\mathbf{A})||}> \int_{0}^{T}B (\tau) \overline{\Gamma}(\tau)d\tau, \nonumber
\end{eqnarray}
for all $i$ and $k$. 

From the Perron-Frobenius theorem for non-negative and irreducible matrices, the leading eigenvalue, $\lambda_1(\mathbf{A})$, of the matrix $\mathbf{A}$ is distinct, real and largest by module, compared to all other eigenvalues; therefore it minimizes the term $1/\beta ||\lambda_i(\mathbf{A}||$. For this reasons, providing: 
\begin{eqnarray}
&&\frac{1}{\beta \lambda_1(\mathbf{A})} > \int_{0}^{T}B (\tau) \overline{\Gamma}(\tau)d\tau \label{thresh_theory}
\end{eqnarray}
holds, the poles of the second set of $N$ equations of the dynamical system (\ref{auxsys}) lie within the left-half of the complex plane, resulting in $\lim_{t \rightarrow \infty} p'^I_i(t) \rightarrow 0$. This yields $lim_{t \rightarrow \infty} p'^E_i(t) \rightarrow 0$, $lim_{t \rightarrow \infty} p^I_i(t) \rightarrow 0$, $lim_{t \rightarrow \infty} p^E_i(t) \rightarrow 0$, and the point of epidemic origin of the dynamical system (\ref{contsys}) is globally asymptotically stable.

The Proof is completed.
\end{proof}

In accordance with the Theorem, the relation: 
\begin{eqnarray}
&&\frac{1}{\beta \lambda_1(\mathbf{A})} = \int_{0}^{T}B (\tau) \overline{\Gamma}(\tau)d\tau \label{thresh_theory_equ}
\end{eqnarray}
defines the boundary between the parametric region related to the state of permanent epidemic presence in the network and the region of epidemic absence. In that sense, the equation (\ref{thresh_theory_equ}) represents the epidemic threshold for the non-Markovian SEIS model occurring on complex networks.

\begin{remark}
The immense importance of proper inclusion of the initial conditions in the non-Markovian SEIS model, is discussed, for the discrete-time case, in details in \cite{nmseis}. The focus of this article is set around the model analysis in the endemic state, i.e. for circumstances in which initial conditions play a minor role. Therefore, it is fairly assumed that the initial epidemic outbreak occurred at moment preceding the beginning of analysis (time $t=0$), for both systems (\ref{disSEIS}) and (\ref{contsys}) and that properly collected set of initial condition, for the time period $(-T,0]$, exists. 
\end{remark}

\section{Equivalence of the stationary states of the non-Markovian SEIS and Markov SIS model}

In this section, we present the main result of the paper -- that the stationary state solutions of the non-Markovian SEIS model, may be found from the stationary state solutions of the Markov SIS model, providing certain relations between process parameters -- infection rates for both Markov SIS and non-Markovian SEIS model $\beta$, curing rate for the Markov SIS model $\gamma$ and the DTPFs (CTPFs) $\overline{\Gamma}(\tau)$ and $B(\tau)$ for the non-Markovian SEIS model -- hold. The purpose of the following analysis is to show that the steady state probabilities that the node $i$ is Exposed, $p_i^E$, for the non-Markovian SEIS model, may be directly related to the steady state probability of the node $i$ being Infected for the Markov SIS model. Knowing steady-state value of $p_i^E$, one may easily calculate the steady state probability of the node $i$ being Infectious, $p_i^I$, for the non-Markovian SEIS model. 

In the following text, we would frequently re-direct the attention of the reader between the Markov SIS and the non-Markovian SEIS model. In order to avoid any confusion in respect to the form we are referring to, in this and in the next Section, labels $M$ for the Markov SIS and $NM$ for the non-Markovian SEIS form would be used for the variables and the parameters (except for $B(\tau)$ and $\overline{\Gamma}(\tau)$), in the form of superscripts. 

\subsection{Discrete-time model}

Starting from the system of equations (\ref{disSEIS}), the stationary state solutions of the discrete-time non-Markovian SEIS model, $p^{E,NM}_{i}=p^{E,NM}_{i}(t)$, $p^{I,NM}_{i}=p^{I,NM}_{i}(t)$, may be found from the relations:
\begin{eqnarray}
p^{E,NM}_i&=&(1-p^{E,NM}_i) \left(1-\prod_{j=1}^N (1-p^{I,NM}_j a_{ij} \beta^{NM})\right)   \sum_{\tau=0}^{T-1} \overline{\Gamma}(\tau)  \label{discstat}\\
p^{I,NM}_i&=& (1-p^{E,NM}_i) \left(1-\prod_{j=1}^N (1-p^{I,NM}_j a_{ij} \beta^{NM})\right) \sum_{\tau=0}^{T-1} B (\tau) \overline{\Gamma} (\tau)  \nonumber
\end{eqnarray}

By dividing equations in (\ref{discstat}):
\begin{eqnarray}
p^{I,NM}_i&=&p^{E,NM}_i \frac{\sum_{\tau=0}^{T-1} B (\tau) \overline{\Gamma} (\tau) }{\sum_{\tau=0}^{T-1} \overline{\Gamma}(\tau)} \label{NMdisc_stat_IvsE}
\end{eqnarray}
From the eq.(\ref{discstat}) and eq.(\ref{NMdisc_stat_IvsE}), the stationary state solution of the non-Markovian SEIS model, in respect to the variable $p^{E,NM}_i(t)$, may be written in the following form:
\begin{eqnarray}
p^{E,NM}_i&=&(1-p^{E,NM}_i) \left(1-\prod_{j=1}^N (1-p^{E,NM}_j a_{ij} \beta^{eff})\right) \sum_{\tau=0}^{T-1} \overline{\Gamma} (\tau)   \label{NM_disc_stat}
\end{eqnarray}
with $\beta^{eff}$ defined with:
\begin{eqnarray}
\beta^{eff}&=&\beta^{NM} \frac{\sum_{\tau=0}^{T-1} B (\tau) \overline{\Gamma} (\tau) }{\sum_{\tau=0}^{T-1} \overline{\Gamma}(\tau)} \label{betaeff}
\end{eqnarray}

Discrete-time Markov SIS model occurring on complex networks is well known in literature, and is mathematically represented in the following form \cite{wang, Chakrabarti}:
\begin{eqnarray}
p^{I,M}_i(t+1)&=&(1-p^{I,M}_i(t)) \left(1-\prod_{j=1}^N (1-p^{I,M}_j  a_{ij} \beta^M)\right)  \label{classSIS}\\
&+&  (1-\gamma^M)p^{I,M}_i(t) \nonumber
\end{eqnarray}
The stationary state solution of the system of equations (\ref{classSIS}), $p^{I,M}_i(t)=p^{I,M}_i$:

\begin{eqnarray}
&&p^{I,M}_i=\frac{(1-p^{I,M}_i) \left(1-\prod_{j=1}^N (1-p^{I,M}_j  a_{ij} \beta^M)\right)}{\gamma^M}  \label{Mar_disc_stat}
\end{eqnarray}

We formulate the stationary state equivalence between the Markov SIS, and non-Markovian SEIS model in a sense that $p^{E,NM}_i=p^{I,M}_i$. By taking:
\begin{eqnarray}
\beta^M=\beta^{eff,NM}=\beta^{NM} \frac{\sum_{\tau=0}^{T-1} B (\tau) \overline{\Gamma} (\tau)}{\sum_{\tau=0}^{T-1} \overline{\Gamma} (\tau)}, \label{eqv_beta_dicrete}
\end{eqnarray}
and then dividing equations (\ref{Mar_disc_stat}) and (\ref{NM_disc_stat}), assuming equality $p^{I,M}_i(t)=p^{I,M}_i$ holds, one obtains:
\begin{eqnarray}
\gamma^{M} = \frac{1}{\sum_{\tau=0}^{T-1} \overline{\Gamma} (\tau)} \label{eqv_gamma_dicrete}
\end{eqnarray}
The analysis conducted above, leads to the following conclusion: providing relations (\ref{eqv_beta_dicrete}) and (\ref{eqv_gamma_dicrete}) hold, the stationary probabilities of each node being Exposed in the discrete non-Markovian SEIS model equals the stationary probabilities of the nodes being in status Infected, in the discrete Markov SIS model.

\subsection{Continuous-time SEIS model}

In the stationary state, $p^{E,NM}_{i}=p^{E,NM}_{i}(t)$, $p^{I,NM}_{i}=p^{I,NM}_{i}(t)$, the relation (\ref{contsys}) takes the following form:
\begin{eqnarray}
&&p^{E,NM}_i=(1-p^{E,NM}_i)   \int_{0}^{T} \overline{\Gamma}(\tau)  d\tau \sum_{j=1}^N a_{ij} \beta^{NM} p^{I,NM}_j  \label{contstat}\\
&&p^{I,NM}_i= (1-p^{E,NM}_i)   \int_{0}^{T} B (\tau) \overline{\Gamma} (\tau) d\tau \sum_{j=1}^N  a_{ij} \beta^{NM} p^{I,NM}_j , \nonumber
\end{eqnarray}

The last pair of equations represent nonlinear system from which one can determine the stationary probabilities. By dividing equations in (\ref{contstat}), We show that in the stationary state, the variables $p^{I,NM}_i$ and $p^{E,NM}_i$ are related in the following fashion:
\begin{eqnarray}
&&p^{I,NM}_i=p^{E,NM}_i \frac{\int_{0}^{T} B (\tau) \overline{\Gamma} (\tau) d\tau }{\int_{0}^{T} \overline{\Gamma} (\tau) d\tau} \label{eq_help}
\end{eqnarray}

If one substitutes the relation (\ref{eq_help}) into the first equation of the system (\ref{contstat}), one obtains:
\begin{eqnarray}
&&p^{E,NM}_i=(1-p^{E,NM}_i) \beta^{NM} \sum_{j=1}^N a_{ij}  p^{E,NM}_j   \int_{0}^{T} B(\tau) \overline{\Gamma}(\tau) d\tau \label{NMstat_sol}
\end{eqnarray}

The Markov SIS model in continuous form is well known \cite{delft}, and may be written as:
\begin{eqnarray}
&&\frac{dp^{I,M}_i(t)}{dt}=(1-p^{I,M}_i(t)) \beta^M\sum_{j=1}^N a_{ij}   p^{I,M}_j(t) -\gamma^M p^{I,M}_i(t) \label{classicalSIScont}
\end{eqnarray}

From (\ref{classicalSIScont}), the stationary state solution, $p^{I,M}_i(t)=p^{I,M}_i$, $dp^{I,M}_i(t)/dt=0$ of the Markov SIS model may be written in the form:
\begin{eqnarray}
&&p^{I,M}_i=\frac{(1-p^{I,M}_i) \beta^M\sum_{j=1}^N a_{ij}   p^{I,M}_j }{\gamma^M} \label{contstatSIS}
\end{eqnarray}

We seek equivalence between the models by equalizing the stationary probability of the arbitrary node being Infected in the Markov SIS model with the stationary probability of the corresponding node being Exposed in the non-Markovian SEIS model, i.e. $p^{I,M}=p^{E,NM}_i$. From relations (\ref{contstatSIS}) and  (\ref{NMstat_sol}) (dividing the two equations under the equality assumption), one obtains that the stationary state of the non-Markovian SEIS spreading process may be obtained form a Markov SIS process providing following relation holds:
\begin{eqnarray}
\frac{\beta^M}{\gamma^M} = \beta^{NM} \int_{0}^{T} B(\tau) \overline{\Gamma}(\tau) d\tau \label{conteqiv}
\end{eqnarray}

One should note the interesting difference between the relations (\ref{eqv_beta_dicrete}, \ref{eqv_gamma_dicrete}) and (\ref{conteqiv}). While (\ref{eqv_beta_dicrete}, \ref{eqv_gamma_dicrete}), in the discrete-time case,  fully define both the $\beta^M$ and $\gamma^M$ for the Markov SIS equivalent, relation (\ref{conteqiv}), in the continuous-time case, leaves certain degree of freedom in the choice of one of these parameters.

As an interesting consequence, one should note that from the second relationship in (\ref{contstat}), when the epidemic is weak $p^{E,NM}_i \approx 0$, one has
\begin{eqnarray}
    &&p_i^{I,NM} = \beta^{NM} \int_{0}^{T} B (\tau) \overline{\Gamma} (\tau) d\tau \sum_{j=1}^N  a_{ij}  p^{I,NM}_j. \nonumber
\end{eqnarray}

In the matrix form, if $\mathbf{P}^{I,NM}$ is the vector of probabilities of infectious state, and $\mathbf{A}$ is the network connectivity matrix, one has
\begin{eqnarray}
    &&\mathbf{P}^{I,NM} = \beta^{NM} \mathbf{A} \mathbf{P}^{I,NM}\int_{0}^{T} B (\tau) \overline{\Gamma} (\tau) d\tau . \nonumber
\end{eqnarray} 

The last expression indicates that the Infectiousness probability vector is eigenvector of scaled connectivity matrix in weak epidemic. This is similar to the previous result that the principal eigenvector determines the probabilities of infection in SEAIR model \cite{basnarkov2021seair}

\subsection{Determination of the epidemic threshold from the Endemic state equivalence}

Standard approach in determining the epidemic threshold for stochastic spreading processes, requires thrall mathematical procedure that is based either on establishing the stability criteria for the system around the point of epidemic origin, as done here in the Section 2.1.2, or complex statistical analysis. In this segment, we show that the endemic model equivalence, enables us to derive the epidemic threshold for non-Markovian SEIS model, directly from the well-known epidemic threshold for the Markov SIS process. 

The epidemic threshold for the Markov SIS process is defined as:
\begin{eqnarray}
&&\lambda_1(\mathbf{A})=\frac{\gamma}{\beta}, \label{thresh_class}
\end{eqnarray}
with $\lambda_1(\mathbf{A})$ being the largest (leading) eigenvalue of the adjacency matrix $\mathbf{A}$. This relation holds for both discrete-time SIS model \cite{wang, Chakrabarti}, as well as continuous-time SIS model \cite{delft, ganesh}.

Stationary-state equivalence between the Markov SIS and non-Markovian SEIS implies that, providing equations (\ref{eqv_beta_dicrete}, \ref{eqv_gamma_dicrete}) in the discrete-time, and (\ref{conteqiv}) in continuous-time SEIS model hold, then:
\begin{eqnarray}
&&p_i^{E,NM}=lim_{t \rightarrow \infty} p_i^{E,NM}(t)=lim_{t \rightarrow \infty} p_i^{I,M}(t)=p_i^{I,M} \nonumber
\end{eqnarray}
Consequently, if $p_i^{I,M}=0$, then $p_i^{E,NM}=0$, as well. Since, $p_i^{I,M}=0$ in the general case, holds for all $i$ if the system is parametersized "under" the epidemic threshold, in accordance with relations (\ref{eqv_beta_dicrete},\ref{eqv_gamma_dicrete}) and (\ref{conteqiv}), the epidemic threshold for the non-Markovian SEIS model is defined with:
\begin{eqnarray}
&&\lambda_1(\mathbf{A})=\frac{\gamma^M}{\beta^M}=\frac{1}{\beta^{NM}  \sum_{\tau=0}^{T-1} B(\tau) \overline{\Gamma}(\tau)}, \label{thresh_dis}
\end{eqnarray}
in the discrete-time case, and:
\begin{eqnarray}
&&\lambda_1(\mathbf{A})=\frac{\gamma^M}{\beta^M}=\frac{1}{\beta^{NM}  \int_{0}^{T} B(\tau) \overline{\Gamma}(\tau) d\tau}, \label{thresh_cont}
\end{eqnarray}
in the continuous case. For the discrete-time case the relation (\ref{thresh_dis}) is derived in precise analytical procedure in \cite{nmseis}. The relation (\ref{thresh_cont}) is identical to the equation (\ref{thresh_theory_equ}) derived in Section 2.1.2.

As presented in this subsection, the stationary-state model equivalence, leads directly to the result for the epidemic threshold of the non-Markovian model. No thrall statistical or system stability analysis is required to obtain this result. This feature further emphasizes the importance of the main result of the paper.

\section{Numerical analysis}

In this section we present the results of the numerical analysis, in order to validate the theoretical results obtained in the previous sections. The analysis is focused around the main result of the paper, i.e. the stationary state equivalence between the non-Markovian SEIS and the Markov SIS model.

The numerical analyzes in the paper are conducted on the following networks:
\begin{itemize}
 \item Barab\'{a}si - Albert \cite{BA} directed and weighted graph with $N=1000$ nodes, total of $L=3992$ uni-directional links , and the largest eigenvalue of the graph's adjacency matrix $\lambda_1(\mathbf{A}) =5.2922$. The reference BA1000, will be used for this network thought the paper. The network is derived from a symmetrical  BA(1000,1996) graph, with $N=1000$ nodes, generated with parameters $m_0=3$, $m=2$; 

 \item Watts - Strogatz \cite{WS} directed and weighted graph with $N=1000$ nodes, total of $L=6000$ uni-directional links, and the largest eigenvalue of the graphs adjacency matrix $\lambda_1(\mathbf{A}) =3.26997$. This network would be further referenced as WS1000. The network is derived from a symmetrical WS(1000,3000) graph, with $N=1000$ nodes, generated with parameters $r=3$, $p=0.2$; 
 \end{itemize}
 
 \subsection{Discrete-time model}
 In order to confirm the results for the discrete-time model, following \cite{nmseis}, we consider two different sets of DTPFs. The first set is related to the case of cumulative-like manifestation of Infectiousness, while in the second set the manifestation has random character. Both sets are presented in Table \ref{Table1}.

\begin{table}[h!]
\caption{DTPFs of the two discrete-time cases cases: Cumulative manifestation and Random manifestation}
\resizebox{\columnwidth}{!}{\begin{tabular}{ |c|c|c|c|c|c|c|c|c|c|c|c|c|}

 \hline
 &$\tau$ & 0 & 1 & 2 & 3 & 4 & 5 & 6 & 7 & 8 & 9 & 10 \\
 \hline
  \multirow{2}{*}{Cumulative manifestation} & $B(\tau)$ &0.0 & 0.0 & 0.0 & 0.1 & 0.2 & 0.3 & 0.5 & 0.8 & 1.0 & 1.0 &1.0\\

& $\Gamma(\tau)$& 0.0 & 0.0 & 0.0 & 0.0 & 0.0 & 0.2 & 0.3 & 0.5 & 0.8 & 0.9 & 1.0\\
 \hline
 \hline
   \multirow{2}{*}{Random manifestation} & $B(\tau)$ & 0.0 & 1.0 & 1.0 & 0.0 & 0.0 &1.0 & 1.0 & 0.0 & 0.0 & 1.0 & 1.0\\

    & $\Gamma(\tau)$ & 0.0 & 0.0 & 0.0 & 0.0 & 0.0 & 0.2 & 0.3 & 0.5 & 0.8 & 0.9 & 1.0\\
 \hline
\end{tabular}}

\label{Table1}
\end{table}

The simulations were conducted as follows: for the non-Markovian SEIS model, for both set of DTPFs, parameter $\beta^{NM}$ was varied in the range $\beta^{NM}\in [0,1)$, with a step of 0.01. Using relations (\ref{eqv_beta_dicrete},\ref{eqv_gamma_dicrete}), parameters $\beta^M$, varied in the parametric region $\beta^M \in [0,0.0.215875]$, in the Cumulative manifestation case, and $\beta^M \in [0,0.495]$  in the Random manifestation case, with $\gamma^M=0.27027$ in both cases.
 
 Results of the analysis in the discrete-time case are presented in Fig. \ref{slika_set4}. 
 
 \begin{figure} [h!]
    \centering
  \subfloat[\label{1a}]{%
       \includegraphics[width=0.48\columnwidth]{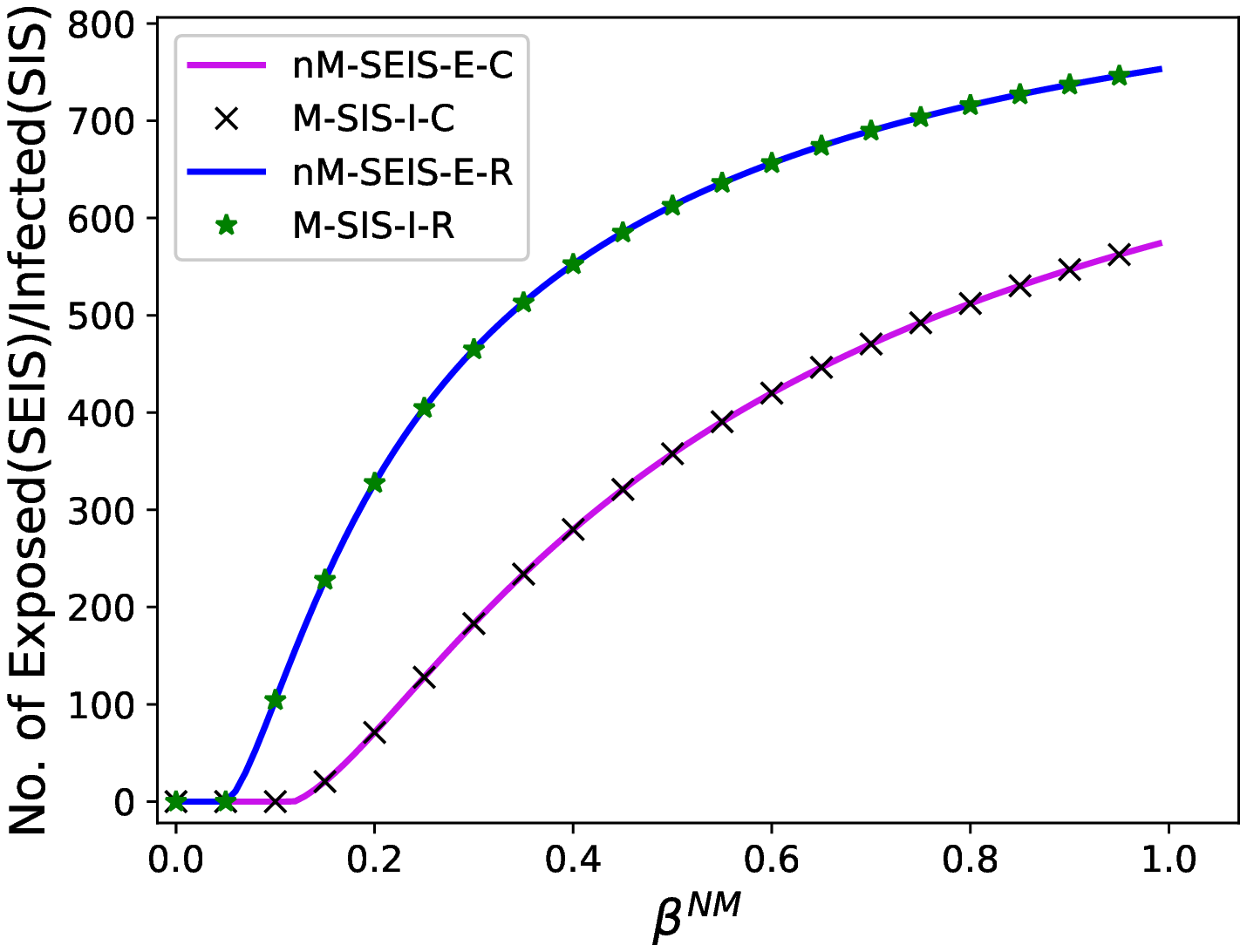}}
    \hfill
  \subfloat[\label{1b}]{%
        \includegraphics[width=0.48\columnwidth]{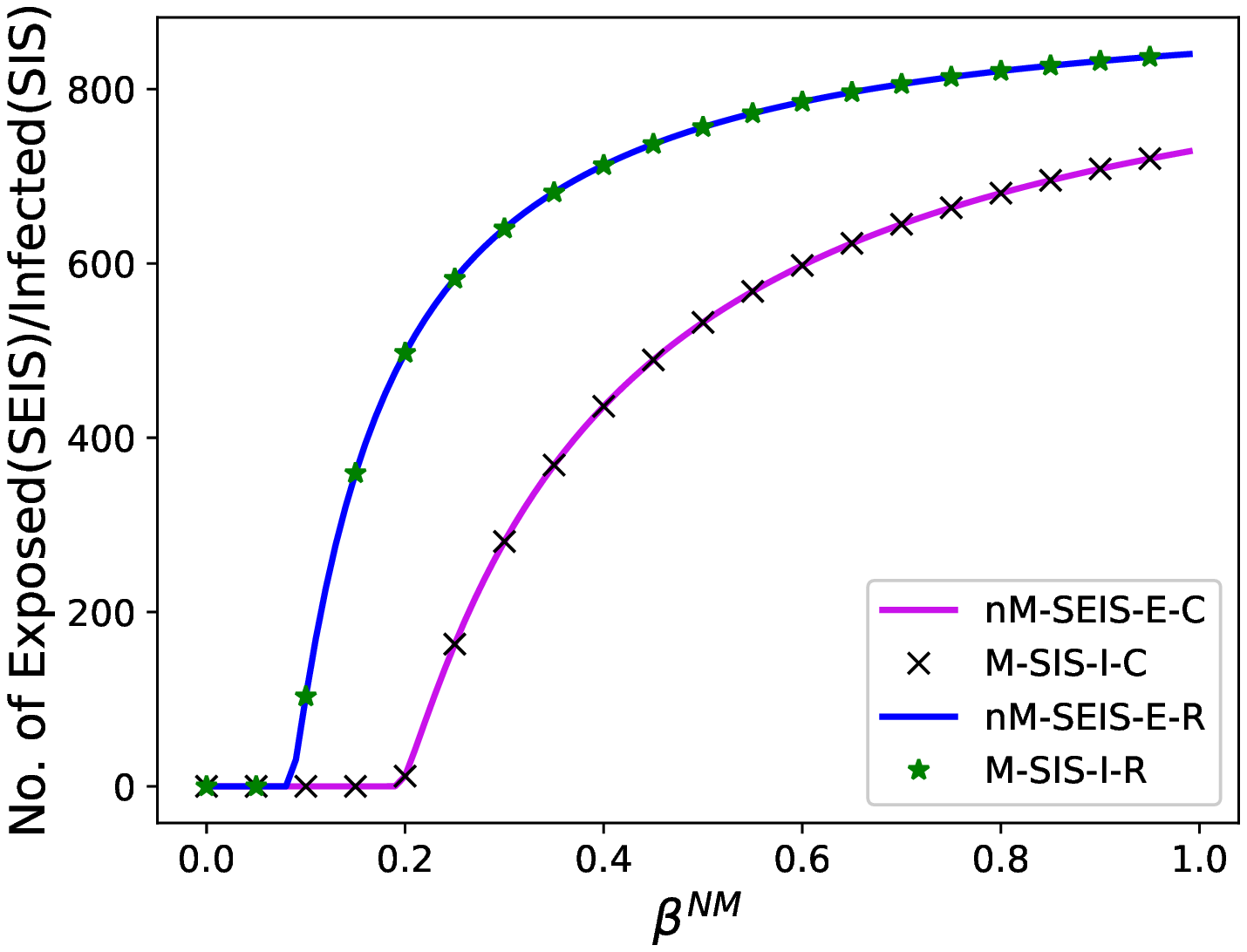}}

  \caption{Comparison between the steady-state values of the number of Exposed individuals in the discrete-time non-Markovian SEIS model and the Markov SIS model. Meaning of the symbols in the legend: nM--non-Markovian, M--Markov; E--Exposed, I--Infected; C--cumulative manifestation, R--random manifestation. (a) BA1000 graph; (b) WS1000 graph.}
  \label{slika_set4} 
\end{figure}

The results from the analysis indicate that there is a perfect overlap of stationary state solutions of both the discrete-time non-Markovian SEIS model and the Markov SIS model, providing relations (\ref{eqv_beta_dicrete},\ref{eqv_gamma_dicrete}) hold.

 \subsection{Continuous-time model}
 The CTPF's used in the continuous-time case (cumulative manifestation only) were constructed as follows: it is assumed that the process lasts for total of $T=65$ time units. The instance manifestation probabilities were obtained from the Weilbull \emph{p.d.f.}: $w(\tau;\alpha;\lambda)=\alpha \lambda (\tau \lambda)^{\alpha-1}\mathrm{exp}\left(-(\tau \lambda)^{\alpha}\right)$, with parameters $\alpha=2.04$ and $\lambda=0.103$ \cite[eq.2]{Qin2020}\cite{Qineabc1202}, normalized to 65 days: $b(\tau)=w(\tau;\alpha;\lambda)/\int_{0}^{65} w(\tau;\alpha;\lambda) d\tau$. The daily recovering probabilities were obtained from log-normal distribution $l(\tau; \mu; \sigma) = 1/(\tau\sigma {\sqrt {2\pi }}) \exp \left(-(\ln \tau-\mu)^{2})/\sigma ^{2}\right)$, $\mu=3$, $\sigma=0.28$ normalized to 60 days, and then time-shifted for 4 days, obtaining: $\gamma(\tau)=0$, for $0 \le \tau \le 4$ and $\gamma(\tau)=l(\tau-4; \mu; \sigma)/\int_{4}^{65}l(\tau-4; \mu; \sigma)d \tau$, for $4 < \tau < 65$. Parameters were chosen to match mean recovery time of $25 \pm 6$ days.
 
 In the analysis, the continuous-time non-Markovian SEIS model was simulated using the integral form (\ref{contsum}), while varying the parameter $\beta^{NM}$ in the $[0,1)$ range. The constant of integration $\Delta \tau = 0.1$ was used in all cases, except for the BA1000 graph with $\beta^{NM}>0.81$, where, for the reasons of numerical stability, the constant of integration was reduced to $\Delta \tau = 0.05$. The Markov SIS model was integrated with forward Euler method, with numerical constant $\Delta t=0.01$. Parameter $\gamma^M=0.01$, was arbitrarily chosen, and parameter $\beta^M$ calculated from (\ref{conteqiv}) varied in the $[0,0.162568]$ region.
 
 Results of the analysis in the continuous-time case are presented in Fig. \ref{slika_cont}. 
\begin{figure} [h!]
\centering
 \includegraphics[width=0.8\columnwidth]{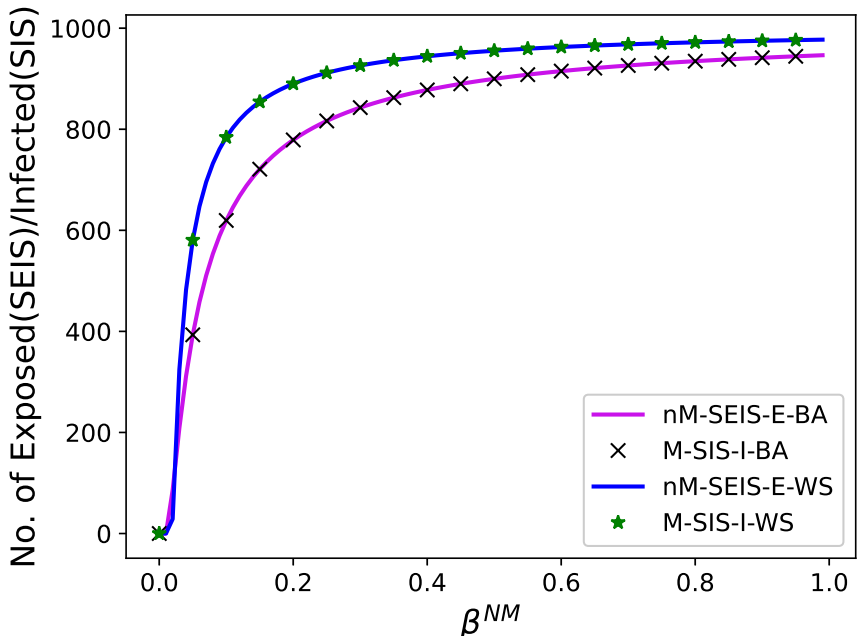}
 \caption{Comparison between the steady-state values of the number of Exposed individuals in the continuous-time non-Markovian SEIS model and the Markov SIS model. Meaning of the symbols in the legend: nM--non-Markovian, M--Markov; E--Exposed, I--Infected; BA--BA1000graph, WS--WS1000graph.}
 \label{slika_cont} 
 \end{figure}

As in the discrete-time case, the results from the analysis indicate that there is a perfect overlap of stationary state solutions of both the continuous-time non-Markovian SEIS model and the Markov SIS model, providing relation (\ref{conteqiv}) holds. However, one should exercise caution when considering the arbitrary choice of one of the parameters $\beta^M$ or $\gamma^M$, in the continuous time case: the application of relation (\ref{conteqiv}), should result in $\beta^M<1$ or $\gamma^M<1$, providing $\gamma^M$ or $\beta^M$ are arbitrary chosen, respectively.

\section{Representation of Markov SIS models as non-Markovian SEIS models}

In this Section, as a result of secondary importance, We show that every classical (Markov) SIS model occurring on complex networks, may be represented as non-Markovian SEIS model, with proper selection of DTPFs/CTPFs $B(\tau)$ and $\gamma(\tau)$. Illustrative numerical example of this feature for the discrete-time case, has already been presented in \cite{nmseis}. In this article, this feature is theoretically investigated and shown for both model forms in a concise mathematical procedure.

As a note to the readers, in what follows, we avoid the use of label superscripts, since the whole procedure is conducted on the non-Markovian SEIS model, that under the investigated circumstances reduces to the Markov SIS form. 

\subsection{Discrete-time case}

Consider, as suggested in \cite{nmseis}, $B(\tau)=1$, for all $\tau$, and $\gamma(\tau)=\gamma (1-\gamma)^{\tau-1} s(\tau-1)$, with $0<\gamma \le 1$, with $s(\tau)$ being the Heaviside function. Consequently, $\Gamma(0)=0$, $\overline{\Gamma}(0)=1$,  $\Gamma(\tau)=1-(1-\gamma)^{\tau}$, $\overline{\Gamma}(\tau)=(1-\gamma)^{\tau}$. Acting similar as in \cite{nmsir}, we obtain:
\begin{eqnarray}
&&p^I_i(t+1)=p^E_i(t+1)=\sum_{\tau=0}^{T-1}\overline{\Gamma} (\tau)(1-p^E_i(t-\tau)) \mathcal{P}_i(t-\tau)=  \nonumber\\
&&=(1-p^E_i(t)) \overline{\Gamma} (0)\mathcal{P}(t) + \sum_{\tau=1}^{T-1}(1-p^E_i(t-\tau))\overline{\Gamma} (\tau) \mathcal{P}_i(t-\tau)= \nonumber\\
&&=(1-p^E_i(t)) \mathcal{P}(t) +(1-\gamma) \sum_{\tau=1}^{T}(1-p^E_i(t-\tau)) \overline{\Gamma} (\tau-1)\mathcal{P}_i(t-\tau) - \nonumber\\ 
&&-(1-p^E_i(t-T))  (1-\gamma)^T \mathcal{P}_i(t-T)= \nonumber\\
&&=(1-p^E_i(t)) \mathcal{P}(t) + (1-\gamma) p^E_i(t) - (1-p^E_i(t-T))  (1-\gamma)^T \mathcal{P}_i(t-T) \nonumber
\end{eqnarray}

When $T \rightarrow \infty$, the last term in the equation vanishes. Bearing in mind that $p^E_i(t)=p^I_i(t)$, the following relation holds:
\begin{eqnarray}
p^I_i(t+1)&=&p^E_i(t+1)=\nonumber\\
&=&(1-p^I_i(t)) \left(1-\prod_{j=1}^N (1-p^I_j(t)  a_{ij} \beta)\right) + (1-\gamma) p^I_i(t) \label{clasSIS}\\
&=& (1-p^E_i(t)) \left(1-\prod_{j=1}^N (1-p^E_j(t)  a_{ij} \beta)\right) + (1-\gamma) p^E_i(t). \nonumber
\end{eqnarray}
The last equation is the equation of the Markov SIS model (\ref{classSIS}). To summarize, for arbitrary values of parameters $\beta$ and $\gamma$, the discrete-time SIS model, may be represented as non-Markovian SEIS model, providing $B(\tau)$ and $\gamma(\tau)$ satisfy the relations defined in the introduction of this Subsection.

\subsection{Continuous-time case}

By considering $B(\tau)=1$, consequently $p^E_i(t)=p^I_i(t)$, and $\overline{\Gamma}(\tau)=exp(-\gamma \tau)$ and $T \rightarrow \infty$ (see \cite{nmsir} for details) one obtains:
\begin{eqnarray}
\frac{dp^I_i(t)}{dt}&=& (1-p^I_i(t)) \sum_{j=1}^N p^I_j(t) a_{ij} \beta -\nonumber\\&-& \gamma \int_{0}^{T} (1-p^I_i(t-\tau)) s(\tau)\sum_{j=1}^N p^I_j(t-\tau) a_{ij} \beta e^{-\gamma \tau}  d\tau\nonumber\\
&=& (1-p^I_i(t)) \sum_{j=1}^N p^I_j(t) a_{ij} \beta \nonumber\\ &-& \gamma \int_{0}^{T} (1-p^I_i(t-\tau)) s(\tau) \sum_{j=1}^N p^I_j(t-\tau) a_{ij} \beta \overline{\Gamma} (\tau) d\tau\nonumber\\
&=& (1-p^I_i(t)) \sum_{j=1}^N p^I_j(t) a_{ij} \beta - \gamma p^I_i(t), \nonumber
\end{eqnarray}
The last relation is identical with the classical formulation of the Markov SIS model occurring on complex network, in continuous time  (\ref{classicalSIScont}). The relation confirms that, for a given Markov SIS, defined by parameters $\beta$ and $\gamma$, exists a non-Markovian SEIS model defined with the identical parameter $\beta$ and CTPFs $B(\tau)=1$ and $\overline{\Gamma}(\tau)=exp(-\gamma \tau)$, such that the non-Markovian SEIS mimics the behaviour of the Markov SIS model on complex networks.

\section{Conclusions}

The non-Markovian systems more accurately address the spreading processes in comparison with Markov models. This characteristic of non-Markovian models originates in their basic definition -- to consider the status transitions that accompany the spreading as non-Poisonous processes, as confirmed by every-day practices. On the adverse side, the numerical analysis of non-Markovian processes is computationally more demanding. The computational complexity increases with the process memory, i.e. with parameter $T$. Even further, in continuous models the choice of an accurate integration step may require substantial amount of computer memory, in order to store and manipulate with an excessive number of preceding states. When these features are accompanied by a huge network, the analysis of a non-Markovian, in this case SEIS, model, may become an overwhelming task for all, but a fairly small number of computing devices. 

In this article we have shown that for the basic non-Markovian re-occurring model, the SEIS model, the stationary state distributions of nodes being Exposed/Infectious, may be found from a Markov SIS equivalent. This result is of at-most importance, since it allows the numerical analysis of the stationary state solutions of the spreading processes to be conducted on systems with standard computational capabilities. In that sense, the result contributes to bringing the computational analysis of non-Markovian models closer to more users, especially individual researchers, small research teams and organizations, that may not be able to acquire a sufficiently powerful computing equipment.

\section*{References}
\bibliographystyle{iopart-num}
\bibliography{mybibfile}

\providecommand{\newblock}{}
\begin{thebibliography}{10}
\expandafter\ifx\csname url\endcsname\relax
  \def\url#1{{\tt #1}}\fi
\expandafter\ifx\csname urlprefix\endcsname\relax\def\urlprefix{URL }\fi
\providecommand{\eprint}[2][]{\url{#2}}

\bibitem{starnini}
Starnini M, Gleeson J~P and Bogu\~n\'a M 2017 {\em Phys. Rev. Lett.\/} {\bf
  118}(12) 128301
  \urlprefix\url{https://link.aps.org/doi/10.1103/PhysRevLett.118.128301}

\bibitem{Nowzari}
{Nowzari} C, {Ogura} M, {Preciado} V~M and {Pappas} G~J 2015 A general class of
  spreading processes with non-markovian dynamics {\em 2015 54th IEEE
  Conference on Decision and Control (CDC)\/} pp 5073--5078

\bibitem{boguna1}
Bogu\~n\'a M, Lafuerza L~F, Toral R and Serrano M~A 2014 {\em Phys. Rev. E\/}
  {\bf 90}(4) 042108
  \urlprefix\url{https://link.aps.org/doi/10.1103/PhysRevE.90.042108}

\bibitem{delft_nm1}
Van~Mieghem P and van~de Bovenkamp R 2013 {\em Phys. Rev. Lett.\/} {\bf
  110}(10) 108701
  \urlprefix\url{https://link.aps.org/doi/10.1103/PhysRevLett.110.108701}

\bibitem{delft_nm2}
Van~Mieghem P and Liu Q 2019 {\em Phys. Rev. E\/} {\bf 100}(2) 022317
  \urlprefix\url{https://link.aps.org/doi/10.1103/PhysRevE.100.022317}

\bibitem{delft_nm3}
Liu Q and Van~Mieghem P 2018 {\em Phys. Rev. E\/} {\bf 97}(2) 022309
  \urlprefix\url{https://link.aps.org/doi/10.1103/PhysRevE.97.022309}

\bibitem{kiss1}
Georgiou N, Kiss I~Z and Scalas E 2015 {\em Phys. Rev. E\/} {\bf 92}(4) 042801
  \urlprefix\url{https://link.aps.org/doi/10.1103/PhysRevE.92.042801}

\bibitem{kiss2}
Kiss I~Z, R\"ost G and Vizi Z 2015 {\em Phys. Rev. Lett.\/} {\bf 115}(7) 078701
  \urlprefix\url{https://link.aps.org/doi/10.1103/PhysRevLett.115.078701}

\bibitem{kiss3}
Sherborne N, Miller J~C, Blyuss K~B and Kiss I~Z 2018 {\em Journal of
  Mathematical Biology\/} {\bf 76} 755--778 ISSN 1432-1416
  \urlprefix\url{https://doi.org/10.1007/s00285-017-1155-0}

\bibitem{Feng2019}
Feng M, Cai S~M, Tang M and Lai Y~C 2019 {\em Nature Communications\/} {\bf 10}
  3748 ISSN 2041-1723
  \urlprefix\url{https://doi.org/10.1038/s41467-019-11763-z}

\bibitem{chun}
Chun J~Y, Baek G and Kim Y 2020 {\em International Journal of Infectious
  Diseases\/} {\bf 99} 403--407 ISSN 1201-9712
  \urlprefix\url{https://www.sciencedirect.com/science/article/pii/S1201971220306123}

\bibitem{Qineabc1202}
Qin J, You C, Lin Q, Hu T, Yu S and Zhou X~H 2020 {\em Science Advances\/} {\bf
  6} (\textit{Preprint}
  \eprint{https://advances.sciencemag.org/content/6/33/eabc1202.full.pdf})
  \urlprefix\url{https://advances.sciencemag.org/content/6/33/eabc1202}

\bibitem{He2020}
He X, Lau E~H~Y, Wu P, Deng X, Wang J, Hao X, Lau Y~C, Wong J~Y, Guan Y, Tan X,
  Mo X, Chen Y, Liao B, Chen W, Hu F, Zhang Q, Zhong M, Wu Y, Zhao L, Zhang F,
  Cowling B~J, Li F and Leung G~M 2020 {\em Nature Medicine\/} {\bf 26}
  672--675 ISSN 1546-170X
  \urlprefix\url{https://doi.org/10.1038/s41591-020-0869-5}

\bibitem{Perelson2002}
Perelson A~S 2002 {\em Nature Reviews Immunology\/} {\bf 2} 28--36 ISSN
  1474-1741 \urlprefix\url{https://doi.org/10.1038/nri700}

\bibitem{nmseis}
Tomovski I, Basnarkov L and Abazi A 2021 {\em IEEE Transactions on Network
  Science and Engineering\/}  1--1 ISSN 2327-4697 (in press)

\bibitem{nmsir}
Basnarkov L, Tomovski I, Sandev T and Kocarev L 2021 {\em arXiv:2107.07427v1
  [physics.soc-ph]\/} \urlprefix\url{https://arxiv.org/abs/2107.07427}

\bibitem{wang}
Wang Y, Chakrabarti D, Wang C and Faloutsos C 2003 Epidemic spreading in real
  networks: An eigenvalue viewpoint {\em In Proc. of the 22nd International
  Symposium on Reliable Distributed Systems - IEEE SRDS'03\/} pp 25--34

\bibitem{Chakrabarti}
Chakrabarti D, Wang Y, Wang C, Leskovec J and Faloutsos C 2008 {\em ACM Trans.
  Inf. Syst. Secur.\/} {\bf 10} 13:1 -- 13:26

\bibitem{gomez1}
G\'{o}mez S, Arenas A, Borge-Holthoefer J, Meloni S and Moreno Y 2010 {\em Eur.
  Phys. Lett.\/} {\bf 89} 38009

\bibitem{gomez2}
G\'{o}mez S, Arenas A, Borge-Holthoefer J, Meloni S and Moreno Y 2011 {\em Int.
  J. Complex Systems in Science\/} {\bf 1} 47--51

\bibitem{delft}
Van~Mieghem P, Omic J and Kooij R 2009 {\em IEEE/ACM Trans. Netw.\/} {\bf 17} 1
  -- 14

\bibitem{basnarkov2021seair}
Basnarkov L 2021 {\em Chaos, Solitons \& Fractals\/} {\bf 142} 110394 ISSN
  0960-0779
  \urlprefix\url{https://www.sciencedirect.com/science/article/pii/S0960077920307876}

\bibitem{ganesh}
Ganesh A~J, Massoulie L and Towsley D~F 2005 The effect of network topology on
  the spread of epidemics {\em In Proc. IEEE Infocom\/} vol~2 pp 1455--1466

\bibitem{BA}
Albert R and Barab\'{a}si A~L 2002 {\em Reviews of Modern Physics\/} {\bf 74}
  47 -- 97

\bibitem{WS}
Watts D Jand~Strogatz S~H 1998 {\em Nature\/} {\bf 393} pp. 440–442

\bibitem{Qin2020}
Qin J, You C, Lin Q, Hu T, Yu S and Zhou X~H 2020 {\em medRxiv : the preprint
  server for health sciences\/}  2020.03.06.20032417 32511426[pmid]
  \urlprefix\url{https://pubmed.ncbi.nlm.nih.gov/32511426}

\end{thebibliography}

\end{document}